\documentclass[11pt]{article}

\usepackage{times}

\usepackage{subcaption}
\usepackage{graphicx}

\def\colorful{0}

\usepackage{times}
\usepackage{amsfonts,amsmath,amssymb,color,float,graphicx,verbatim}
\usepackage{multirow}

\newif\ifhyper\IfFileExists{hyperref.sty}{\hypertrue}{\hyperfalse}
\hyperfalse
\hypertrue
\ifhyper\usepackage{hyperref}\fi

\usepackage{enumitem}
\usepackage[capitalize]{cleveref} 
\usepackage{algorithm}

\def\nnewcolor{1}
\ifnum\nnewcolor=1

\fi
\ifnum\nnewcolor=0

\fi

\ifnum\colorful=1
\newcommand{\new}[1]{{\color{red} #1}}

\else
\newcommand{\new}[1]{{#1}}

\fi

\usepackage{amsthm,amsfonts,amsmath,color,float,graphicx,verbatim}
\usepackage{multirow}
\usepackage{amsmath, color, enumitem}
\usepackage{framed}
\usepackage{nicefrac}

\newtheorem{theorem}{Theorem}[section]

\newtheorem{lemma}[theorem]{Lemma}
\newtheorem{proposition}[theorem]{Proposition}

\newtheorem{fact}[theorem]{Fact}

\newtheorem{remark}[theorem]{Remark}

\newtheorem{definition}[theorem]{Definition}

\newcommand{\vol}{{\mathrm{vol}}}
\newcommand{\R}{\mathbb{R}}

\newcommand{\Z}{\mathbb{Z}}
\newcommand{\E}{\mathbb{E}}

\newcommand{\poly}{\mathrm{poly}}
\newcommand{\polylog}{\mathrm{polylog}}

\newcommand{\dtv}{d_{\mathrm TV}}

\newcommand{\ignore}[1]{}

\newcommand{\eps}{\epsilon}

\newcommand{\Poi}{\mathrm{Poi}}

\newcommand{\etens}[1]{\otimes{#1}}

\newcommand{\eqdef}{\stackrel{{\mathrm {\footnotesize def}}}{=}}

\title{Testing Identity of Multidimensional Histograms}

\author{Ilias Diakonikolas
\thanks{Supported by NSF Award CCF-1652862 (CAREER) and a Sloan Research Fellowship.}\\
University of Southern California\\
{\tt iliasdiakonikolas@gmail.com}\\
\and
Daniel M. Kane\thanks{Supported by NSF Award CCF-1553288 (CAREER) and a Sloan Research Fellowship.}\\
University of California, San Diego\\
{\tt dakane@cs.ucsd.edu}\\
\and
John Peebles
\thanks{Supported by the NSF Graduate Research Fellowship under Grant 1122374, and by the NSF Grant 1065125. 
Some of this work was performed while visiting USC.}\\
CSAIL, MIT \\
{\tt jpeebles@mit.edu}\\
}

\begin{document}

\maketitle

\setcounter{page}{0}

\thispagestyle{empty}

\begin{abstract}
We investigate the problem of identity testing for multidimensional histogram distributions.
A distribution $p: D \to \R_+$, where $D \subseteq \R^d$, is called a {$k$-histogram} if there exists a partition of the
domain  into $k$ axis-aligned rectangles such that $p$ is constant within each such rectangle.
Histograms are one of the most fundamental nonparametric families of distributions and
have been extensively studied in computer science and statistics.
We give the first identity tester for this problem with {\em sub-learning} sample complexity
in any fixed dimension and a nearly-matching sample complexity lower bound.

In more detail, let $q$ be an unknown $d$-dimensional $k$-histogram distribution in fixed dimension $d$, 
and $p$ be an explicitly given $d$-dimensional $k$-histogram.
We want to correctly distinguish, with probability at least $2/3$, between the case that $p = q$ versus $\|p-q\|_1 \geq \eps$.
We design an algorithm for this hypothesis testing problem
with sample complexity $O((\sqrt{k}/\eps^2) 2^{d/2} \log^{2.5 d}(k/\eps))$ that runs in sample-polynomial time.
Our algorithm is robust to model misspecification, i.e., succeeds even if $q$ is only promised 
to be {\em close} to a $k$-histogram.
Moreover, for $k = 2^{\Omega(d)}$, we show a sample complexity lower bound of
$(\sqrt{k}/\eps^2) \cdot \Omega(\log(k)/d)^{d-1}$ when $d\geq 2$. 
That is, for any fixed dimension $d$, our upper and lower bounds are nearly matching.
Prior to our work, the sample complexity of the $d=1$ case was well-understood,
but no algorithm with sub-learning sample complexity was known, even for $d=2$. Our new upper and lower bounds
have interesting conceptual implications regarding the relation between learning and testing in this setting.
\end{abstract}

\newpage

\section{Introduction} \label{sec:intro}

\subsection{Background} \label{ssec:background}

The task of verifying the identity of a statistical model --- known as {\em identity testing} or {\em goodness of fit} ---
is one of the most fundamental questions in statistical hypothesis testing~\cite{Pearson1900, NeymanP}.
In the past two decades, this question has been extensively studied by the TCS and information-theory communities
in the framework of {\em property testing}~\cite{RS96, GGR98}:
Given sample access to an unknown distribution $q$ over a finite domain $[n]: = \{1, \ldots, n\}$, an explicit distribution
$p$ over $[n]$, and a parameter $\eps>0$, we want to distinguish between the cases that $q$ and $p$ are identical versus
$\eps$-far from each other in $\ell_1$-norm (statistical distance).
Initial work on this problem focused on characterizing the sample size needed to test the identity of an arbitrary
distribution of a given support size $n$. This regime is well-understood: there exists an efficient
estimator with sample complexity $O(\sqrt{n}/\eps^2)$~\cite{VV14, DKN:15, ADK15}
that is worst-case optimal up to constant factors.

The aforementioned sample complexity characterizes worst-case instances
and drastically better upper bounds may be possible if we have some a priori
qualitative information about the unknown distribution. For example, if $q$
is an {\em arbitrary} continuous distribution, no identity tester with finite sample complexity exists.
On the other hand, if $q$ is known to have some nice structure, the domain size may
not be the right complexity measure for the identity testing problem
and one might hope that strong positive results can be obtained even for the continuous setting.
This discussion motivates the following natural question: {\em To what extent can we exploit the underlying structure
to perform the desired statistical estimation task more efficiently?}

A natural formalization of the aforementioned question involves
assuming that the unknown distribution belongs to (or is close to) a given family of
distributions. Let $\mathcal{D}$ be a family of distributions over $\R^d$.
The problem of {\em identity testing for $\mathcal{D}$} is the following:
Given sample access to an unknown distribution $q \in \mathcal{D}$, and an explicit distribution $p \in \mathcal{D}$,
we want to distinguish between the case that $q = p$ versus $\|q-p\|_1 \ge \eps.$
(Throughout this paper, $\|p-q\|_1$ denotes the $L_1$-distance between the distributions $p, q$.)
We note that the sample complexity of this testing problem depends on the complexity of the underlying class
$\mathcal{D}$, and it is of fundamental interest to obtain efficient algorithms
that are {\em sample optimal} for $\mathcal{D}$. A recent body of work in distribution testing has focused on
leveraging such a priori structure to obtain significantly improved sample
complexities~\cite{BKR:04, DDSVV13, DKN:15, DKN:15:FOCS, CDKS17, DaskalakisP17, DaskalakisDK16, DKN17}.

One approach to solve the identity testing problem for a family $\mathcal{D}$ is to learn $q$ up to $L_1$-distance $\eps/3$
and then check (without drawing any more samples) 
whether the hypothesis is $\eps/3$-close to $p$. Thus, the sample complexity of identity testing for $\mathcal{D}$
is bounded from above by the sample complexity of {\em learning} (an arbitrary distribution in) $\mathcal{D}$.
It is natural to ask whether a better sample size bound could be achieved for the identity testing
problem, since this task is, in some sense,  less demanding than the task of learning. In this paper, we provide
an affirmative answer to this question for the family of multidimensional histogram distributions.

\subsection{Our Results: Identity Testing for Multidimensional Histograms} \label{ssec:results}

In this work, we investigate the problem of identity testing for multidimensional histogram distributions.
A $d$-dimensional probability distribution with density $p: D \to \R$, where $D \subset \R^d$ is either $[m]^d$ or $[0, 1]^d$,
is called a \emph{$k$-histogram} if there exists a partition of the domain
into $k$ axis-aligned rectangles $R_1, \ldots, R_k$ such that $p$ is constant on $R_i$, for all $i = 1,\ldots, k$.
We let $\mathcal{H}^d_k (D)$ denote the set of $k$-histograms over $D$.
We will use the simplified notation $\mathcal{H}^d_k$ when the underlying domain
is clear from the context.
Histograms constitute one of the most basic nonparametric distribution families and
have been extensively studied in statistics and computer science.

Specifically, the problem of learning histogram distributions from samples has been extensively studied
in the statistics community and many methods have been proposed~\cite{Scott79, FreedmanD1981, 
Scott:92, LN96, Devroye2004, WillettN07, Klem09}
that unfortunately have a strongly exponential dependence on the dimension.
In the database community, histograms~\cite{JPK+98,CMN98,TGIK02,GGI+02, GKS06, ILR12, ADHLS15}
constitute the most common tool for the succinct approximation of large datasets.
Succinct multivariate histograms representations are well-motivated in several data analysis
applications in databases, where randomness is used to subsample a large dataset~\cite{CGHJ12}.

In recent years, histogram distributions have attracted renewed interested
from the TCS community
in the context of learning~\cite{DDS12soda, CDSS13, CDSS14, CDSS14b, 
DHS15, AcharyaDLS16, ADLS17, DiakonikolasKS16a, DLS18} and
testing~\cite{ILR12, DDSVV13, DKN:15,  DKN:15:FOCS, Canonne16, CDGR16, DKN17}.
The algorithmic difficulty in learning and testing such distributions lies in the fact
that the location and size of the rectangle partition is unknown.
The majority of the literature has focused on the univariate setting
which is by now well-understood. Specifically, it is known that the sample complexity of learning
$\mathcal{H}^1_k$ is $\Theta(k/\eps^2)$ (and this sample bound is achievable with
computationally efficient algorithms~\cite{CDSS14, CDSS14b, ADLS17}); while the sample complexity of
identity testing $\mathcal{H}^1_k$ is $\Theta(\sqrt{k}/\eps^2)$~\cite{DKN:15}. That is, in one dimension, the gap between learning and
identity testing as a function of the complexity parameter $k$ is known to be quadratic.

A recent work~\cite{DLS18} obtained a sample near-optimal and computationally
efficient algorithm for {\em learning} multidimensional $k$-histograms in any fixed dimension.
The sample complexity of the~\cite{DLS18} algorithm is
$O((k/\eps^2) \log^{O(d)} (k / \eps))$ while the optimal sample complexity
of the learning problem (ignoring computational considerations) is 
$\widetilde{\Theta}(d k / \eps^2)$\footnote{We note that the $\widetilde{\Theta}()$ notation 
hides polylogarithmic factors in its argument.}.
On the other hand, the property testing question in two (or more) dimensions
is poorly understood. In particular, prior to this work, no testing algorithm with sub-learning sample complexity
was known, even for $d=2$ (independent of computational considerations). 
In this paper, we obtain an identity tester for multidimensional histograms in any fixed dimension 
with {\em sub-learning sample complexity} and establish a nearly-matching sample complexity lower bound 
(that applies even to the special case of uniformity testing). Our main result is the following:

\begin{theorem}[Main Result] \label{thm:main}
Let $\eps > 0$ and $k \in \mathbb{Z}_+$.
Let $q \in \mathcal{H}^d_k(D)$ be an unknown $k$-histogram distribution over $D =[0, 1]^d$  
or $D =[m]^d$, \new{where $d$ is fixed}, and $p \in \mathcal{H}^d_k(D)$ be explicitly given.
There is an algorithm which draws $m = O((\sqrt{k}/\eps^2) 2^{d/2} \log^{2.5 d}(dk/\eps))$ samples from $q$, 
runs in sample-polynomial time, and distinguishes, with probability at least $2/3$, 
between the case that $p = q$ versus $\|p-q\|_1 \geq \eps$.
Moreover, any algorithm for this hypothesis testing problem requires
$(\sqrt{k}/\eps^2) \Omega(\log(k)/d)^{d-1}$ samples for $k = 2^{\Omega(d)}$, even for uniformity testing.
\end{theorem}

A few remarks are in order:
First, we emphasize that the focus of our work is on the case where the parameter $k$ is
{\em much larger} than the dimension $d$. For example, this condition is automatically satisfied when
$d$ is bounded from above by a fixed constant. We note that understanding the regime of fixed dimension $d$ 
is of fundamental importance, as it is the most commonly studied setting in nonparametric inference. Moreover,
in several of the classical database and streaming applications of multidimensional histograms 
(see, e.g.,~\cite{Poosala97, GunopulosKTD00, BrunoCG01, Muthu05} and references therein) 
the dimension $d$ is relatively small (at most $10$), while the number of rectangles is orders of magnitude larger .
For such parameter regimes, our identity tester
has sub-learning sample complexity that is near-optimal, up to the precise power of the logarithm 
(as follows from our lower bound). Understanding the parameter regime where $k$ and $d$ are comparable, 
e.g., $k = \poly(d)$, is left as an interesting open problem.

It is important to note that our identity testing algorithm is {\em robust} to model misspecification. Specifically,
the algorithm is guaranteed to succeed as long as the unknown distribution $q$ is $\eps/10$-close, in $L_1$-norm, to
being a $k$-histogram. This robustness property is important in applications and is conceptually interesting
for the following reason: In high-dimensions, robust identity testing with sub-learning sample complexity is provably impossible,
even for the simplest high-dimensional distributions, including spherical Gaussians~\cite{DiakonikolasKS16c}.

A conceptual implication of Theorem~\ref{thm:main} concerns the sample complexity gap between
learning and identity testing for histograms. It was known prior to this work that
the gap between the sample complexity of learning and identity testing for {\em univariate} $k$-histograms
is quadratic as a function of $k$. Perhaps surprisingly, our results imply that
this gap decreases as the dimension $d$ increases (as long as the dimension remains fixed). 
This follows from our sample complexity lower bound in Theorem~\ref{thm:main} and the fact that the sample complexity of
learning $\mathcal{H}^d_k$ is $\tilde{\Theta}(d k/\eps^2)$ 
(as follows from standard VC-dimension arguments, see, e.g.,~\cite{DLS18}).
In particular, even for $d=3$, the gap between the sample complexities of learning and identity testing
is already {\em sub-quadratic} and continues to decrease as the dimension increases.
(We remind the reader that our lower bound applies for $k>2^{\Omega(d)}$.)

Finally, we note here a {\em qualitative} difference between the $d=1$ and $d \geq 2$ cases.
Recall that for $d=1$ the sample complexity of identity testing $k$-histograms is $\Theta(\sqrt{k}/\eps^2)$.
For $d=2$, the sample complexity of our algorithm is $O((\sqrt{k}/\eps^2) \log^5(k/\eps))$. It would be tempting
to conjecture that the multiplicative logarithmic factor is an artifact of our algorithm and/or its analysis. Our lower bound
of $\Omega ((\sqrt{k}/\eps^2) \log(k))$ shows that some constant power of a logarithm is in fact necessary.

\subsection{Related Work}

The field of \emph{distribution property testing}~\cite{BFR+:00} has been extensively investigated in the past
couple of decades, see~\cite{Rub12, Canonne15,Gol:17}. 
A large body of the literature has focused on characterizing the sample size needed to test properties
of arbitrary discrete distributions. This regime is fairly well understood:
for many properties of interest there exist sample-efficient testers
~\cite{Paninski:08, CDVV14, VV14, DKN:15, ADK15, CDGR16, DK16, DiakonikolasGPP16, CDS17, Gol:17, DGPP17, BatuC17, DKS17-gu, CDKS18}.
More recently, an emerging body of work has focused on leveraging \textit{a priori} structure
of the underlying distributions to obtain significantly improved sample
complexities~\cite{BKR:04, DDSVV13, DKN:15, DKN:15:FOCS, CDKS17, DaskalakisP17, DaskalakisDK16, DKN17}.

The area of distribution inference under structural assumptions ---  that is, inference about a distribution
under the constraint that its probability density function satisfies certain qualitative properties ---
is a classical topic in statistics starting with the pioneering work of Grenander~\cite{Grenander:56}
on monotone distributions. The reader is referred to~\cite{BBBB:72} for a summary of the early work
and to~\cite{GJ:14} for a recent book on the subject.
This topic is well-motivated in its own right, and
has seen a recent surge of research activity in the statistics and econometrics communities,
due to the ubiquity of structured distributions in the sciences. The conventional wisdom is that,
under such structural constraints,  the quality of the resulting estimators may dramatically improve,
both in terms of sample size and in terms of computational efficiency.

\subsection{Basic Notation} \label{ssec:defs}
We will use $p, q$ to denote the probability density functions (or probability mass functions)
of our distributions. If $p$ is discrete over support $[n] \eqdef \{1, \ldots, n\}$, we denote
by $p_i$ the probability of element $i$ in the distribution.
For discrete distributions $p, q$, their $\ell_1$ and $\ell_2$ distances are
$\|p -q \|_1 = \sum_{i=1}^n |p_i - q_i|$ and $\|p-q\|_2 = \sqrt{\sum_{i=1}^n (p_i - q_i)^2}$.
For $D \subseteq \R^d$ and density functions $p, q: D \to \R_+$, we have $\|p -q \|_1 = \int_D |p(x)-q(x)| dx$.
The total variation distance between distributions $p, q$ is defined to be $\dtv(p, q) = (1/2) \cdot \|p -q \|_1$.

Fix a partition of the domain $D$ into disjoint sets
$\mathcal{S} :=  (S_i)_{i=1}^{\ell}.$ For such a partition $\mathcal{S}$,
the {\em reduced distribution} $p_r^{\mathcal{S}}$ corresponding to $p$ and $\mathcal{S}$ 
is the discrete distribution over $[\ell]$ that assigns the $i$-th ``point'' the mass that $p$ assigns to the
set $S_i$; i.e., for $i \in [\ell]$, $p_r^{\mathcal{S}} (i) = p(S_i)$.

Our lower bound proofs will use the following metric, which can be seen as a generalization of 
the chi-square distance: For probability distributions $p, q$ and $r$ let
$
\chi_p(q,r) \eqdef \int \frac{dq dr}{dp} \;.
$




\subsection{Overview of Techniques}

In this section, we provide a high-level overview of our algorithmic and lower bounds techniques
in tandem with a comparison to prior related work.

\paragraph{Overview of Identity Testing Algorithm}

We start by describing our {\em uniformity tester} for $d$-dimensional $k$-histograms. 
For the rest of this intuitive description, we focus on histograms over $[0, 1]^d$.
A standard, yet important, tool we will use is the concept of a \emph{reduced distribution} defined above. 
Note that a random sample from the reduced distribution $p_r^{\mathcal{S}}$ can be obtained 
by taking a random sample from $p$ and returning the element of the partition 
that contains the sample.

The first observation is that if the unknown distribution $q \in \mathcal{H}^d_k$
and the uniform distribution $p = U$ are $\eps$-far in $L_1$-distance,
there exists a partition of the domain into $k$ rectangles $R_1, \ldots, R_k$
such that the difference between $q$ and $p$ can be detected based on the reduced distributions on this partition.
If we knew the partition $R_1, \ldots, R_k$ ahead of time, the testing problem would be easy:
Since the reduced distributions have support $k$, this would yield a uniformity tester
with sample complexity $O(\sqrt{k}/\eps^2)$.
The main difficulty is that the correct partition is unknown to the testing algorithm
(as it depends on the unknown histogram distribution $q$).

A natural approach, employed in~\cite{DKN:15} for $d=1$, is to appropriately ``guess'' the correct rectangle partition.
For the univariate case, a {\em single} interval partition already leads to a non-trivial uniformity tester.
Indeed, consider partitioning the domain into $\Theta(k/\eps)$ intervals of equal length (hence, of equal mass
under the uniform distribution). It is not hard to see that the reduced distributions over these intervals
can detect the discrepancy between $q$ and $p$, leading to a uniformity tester with sample complexity
$\Theta((k/\eps)^{1/2}/\eps^2) = \Theta(k^{1/2}/\eps^{5/2})$. This very simple scheme gives an identity testing with sub-learning
sample complexity when $\eps$ is constant --- albeit suboptimal for small $\eps$. Unfortunately, such an approach
can be seen to inherently fail even for two dimensions: Any {\em obliviously chosen} partition in two dimensions
requires $\Omega(k^2/\eps^2)$ rectangles, which leads to an identity tester with
sample complexity $\Omega(k/\eps^3)$. Hence, a more sophisticated approach is required in two dimensions
to obtain {\em any} improvement over learning.

Instead of using a single oblivious interval decomposition of the domain,
the sample-optimal $\Theta(k^{1/2}/\eps^{2})$ uniformity tester of~\cite{DKN:15} for univariate $k$-histograms
partitions the domain into intervals in {\em several different ways}, and runs a known $\ell_2$-tester
on the reduced distributions (with respect to the intervals in the partition) as a black-box.
At a high-level, we appropriately generalize this idea to the multidimensional setting.

To achieve this, we proceed by partitioning the domain into approximately $k$ identical
rectangles, distinguishing the different partitions based on the
shapes of these rectangles. This requirement to guess the shape is
necessary, as for example partitioning the square into rows will not
suffice when the true partition is a partition into columns. 
We show that it suffices to consider a poly-logarithmic sized set of partitions, where
any desired shape of rectangle can be achieved to within a factor of
$2$. In particular, we show that for each of the $k$ rectangles in the
true partition that are sufficiently large, at least one of our
oblivious partitions will use rectangles of approximately the same
size, and thus at least one rectangle in this partition will
approximately capture the discrepancy due to this rectangle (note that
only considering large rectangles suffices, since any rectangle on
which the uniform distribution assigns substantially more mass than $q$
must be reasonably large). This means that at least one partition will
have an $\eps/\polylog(k/\eps)$ discrepancy between $p$ and $q$, and by running an
identity tester on this partition, we can distinguish them.

One complication that arises here is that for small values of $\eps$, the
difference between $p$ and $q$ might be due to rectangles with area much
less than $1/k$. In order to capture these rectangles, we will need some
of our oblivious partitions to be into rectangles with area smaller
than $1/k$, for which there will necessarily be more than $k$ rectangles
in the partition (in fact, as many as $k/\eps$ many rectangles).
This would appear to cause problems for the following reason:
the sample complexity of $\ell_1$-uniformity testing over a discrete domain of size $n$
is $\Theta(n^{1/2}/\eps^2)$. Hence, naively using such a uniformity tester
on the reduced distributions obtained by a decomposition into $k/\eps$
rectangles would lead to the sub-optimal sample complexity of 
$\Theta((k/\eps)^{1/2}/\eps^2) = \Theta(k^{1/2}/\eps^{5/2})$.

We can circumvent this difficulty by leveraging the following insight:
Even though the total number of rectangles in the partition might be large, it can be
shown that for a well-chosen oblivious partition, {\em a reasonable
fraction of this discrepancy is captured by only $k$ of these rectangles}.
In such a case, the sample complexity of uniformity testing can be notably
reduced using an ``$\ell_1^k$-identity tester'' --- an identity tester
under a {\em modified metric} {\em that measures the discrepancy of the largest
$k$ domain elements}. \new{By leveraging the flattening method of~\cite{DK16},} 
we design such a tester with the optimal sample complexity of 
$O(\sqrt{k}/\eps^2)$ (Theorem~\ref{thm:l1k}) --- independent of the domain size.
This completes the sketch of our uniformity tester for
the multidimensional case.

To generalize our uniformity tester to an identity tester for multidimensional histograms,
two significant problems arise. The first is that it is no longer clear what the shape
of rectangles in the oblivious partition should be. This is because when the explicit distribution $p$
is not the uniform distribution, equally sized rectangles are not a natural option
to consider. This problem can be fixed by breaking the axes
into pieces that assign equal mass to the marginals of the known distribution 
(Lemma~\ref{lem:oblivious-covering-construction}).
The more substantial problem is that it is no longer clear that the discrepancy
between $p$ and $q$ can be captured by a partition of the square into $k$ rectangles.
This is because the two $k$ rectangle partitions corresponding to the $k$-histograms
$p$ and $q$ when refined could lead to a partition of the square into as many $k^2$ rectangles.

To remedy this, we note that there is still a partition into $k$ rectangles such that $q$
is piecewise constant on that partition. We show (Lemma~\ref{lem:split-capturing}) that if we refine this partition
slightly --- by dividing each region into two regions, the half on which $p$ is heaviest
and the half on which $p$ is lightest --- this new partition will
capture a constant fraction of the difference between $p$ and $q$. Given this structural result,
our identity testing algorithm becomes similar to our uniformity tester. We
obliviously partition our domain into rectangles poly-logarithmically
many times, each time we now divide each rectangle further into two
regions as described above, and then run identity testers on these
partitions. We show that if $p$ and $q$ differ by $\eps$ in $L_1$-distance,
then at least one such partition will detect at least $\eps/\polylog(k/\eps)$ of this
discrepancy.

\paragraph{Overview of Sample Complexity Lower Bound}
Note that $\Omega(\sqrt{k}/\eps^2)$ is a straightforward lower bound
on the sample complexity of identity testing $k$-histograms, even for $d=1$.
This follows from the fact that a $k$-histogram can simulate any discrete distribution over $k$ elements.

In order to prove lower bounds of the form $\omega(\sqrt{k}/\eps^2)$,
we need to show that any tester {\em must} consider many possible shapes of
rectangles. This suggests a construction where we have a grid of some
unknown dimensions, where some squares in the grid are dense and
the remainders are sparse in a checkerboard-like pattern. It should be
noted that if we have two such grids whose dimensions differ by
exactly a factor of $2$, it can be arranged such that the distributions are exactly uncorrelated
with each other. Using this observation, we can construct $\log(k)$ such uncorrelated
distributions that the tester will need to check for individually.
Unfortunately, this simple construction will not suffice to prove our desired lower bound, 
as one could merely run $\log(k)$ different testers in parallel. (We note, however, that this construction
does yield a non-trivial lower bound, see Proposition~\ref{IntermediateLBProp}.) 
We will thus need a slightly more elaborate construction, which we now describe: 
First, we divide the square domain into $\polylog(k)$ equal regions. Each of these regions is turned
into one of these randomly-sized checkerboards, but where different
regions will have different scales. We claim that this ensemble is
hard to distinguish from the uniform distribution.

The formal proof of the above sketched lower bound is somewhat technical and involves
bounding the chi-squared distance of taking $\Poi(m)$ samples from a random
distribution in our ensemble with respect to the distribution obtained
by taking $\Poi(m)$ samples from the uniform distribution. 

Bounding the chi-square distance is simplified by noting that since the sets of samples
from each of the $\sqrt{k}$ bins are independent of each other, we can
consider each of them independently. For each individual bin, we take
$s \sim \Poi(m/\sqrt{k})$ samples and need to compute
$\chi_{U^{\etens s}}(X^{\etens s},Y^{\etens s}) = \chi_U(X,Y)^{\new{s}}$, where $X$ and $Y$
are random distributions from our ensemble and $U$ is the uniform distribution. 
It is not hard to see that if $X$ and $Y$ are checkerboards of different scales,
then the $\chi^2$-value is exactly $1$.
This saves us a factor of $\log(k)$, as there are $\log(k)$ many different
scales to consider, and leads to the desired 
sample lower bound (Theorem~\ref{finalLBThm}).

\medskip

\noindent {\bf Organization} In Section~\ref{sec:algo}, we give our identity
testing algorithm. Our sample complexity lower bound proof 
is given in Section~\ref{sec:lb}. Finally, Section~\ref{sec:conc} outlines
some directions for future work.

\section{Sample Near-Optimal Identity Testing Algorithm} \label{sec:algo}

In Section~\ref{ssec:algo-proof},
we describe and analyze our identity tester, assuming the existence of a good
oblivious covering. In Section~\ref{ssec:good-cover}, we show the existence of such a covering.

\subsection{Algorithm and its Analysis} \label{ssec:algo-proof}

Let $q$ be the unknown histogram distribution and $p$
be the explicitly known one.
Our algorithm considers several judiciously chosen oblivious decompositions of the domain
that will be able to approximate a set on which we can distinguish our distributions.
We formalize the properties that we need these decompositions to have with the notion
of a {\em good oblivious covering} (Definition~\ref{defn:oblivious-covering} below).
The essential idea is that we cover the domain $[0,1]^d$ with rectangles that do not
overlap too much in such a way so that any partition of $[0,1]^d$ into $k$ rectangles
can be approximated by some union of rectangles in this family.

\begin{definition}[good oblivious covering]\label{defn:oblivious-covering}
Let $p$ be a probability distribution on $[0,1]^d$.
For $k, j, \ell \in \Z_+$ and $0<\eps \leq 1/2$,
a $(k, j, \ell, \eps)$-\emph{oblivious covering} of $p$ is a family $\mathcal{F}$ of subsets of $[0,1]^d$
satisfying the following:
\begin{enumerate}
\item For any partition $\Pi$ of $[0,1]^d$ into $k$ rectangles, there exists a subfamily $\mathcal{S} \subseteq \mathcal{F}$
such that:
\begin{enumerate}
\item We have that $|\mathcal{S}| \leq k \cdot j$ \label{item:size}.
\item The sets in $\mathcal{S}$ are mutually disjoint, i.e., $S_1 \bigcap S_2 = \emptyset$ for all $S_1 \neq S_2 \in \mathcal{S}$.
\item The sets in $\mathcal{S}$ together contain all except at most $\eps$ of the probability mass
of  $[0,1]^d$ under $p$, i.e., $p(\cup_{S \in \mathcal{S}} S) \geq 1-\eps$. \label{item:coverage}
\item For each $S \in \mathcal{S}$ there is some histogram rectangle
$R \in \Pi$ such that $S$ only contains points from $R$, i.e., $S \subseteq R$. \label{item:respect}
\end{enumerate}
\item For each point $x$ in $[0,1]^d$, the number of sets in $\mathcal{F}$ containing $x$ is exactly $\ell$.
\end{enumerate}
\end{definition}

\newcommand{\finelen}{\ensuremath{4kd/\epsilon}}
In Section~\ref{ssec:good-cover}, Lemma~\ref{lem:oblivious-covering-construction}, 
we establish the existence of a $(k,2^d \log^{d}(\finelen),\log^{d}(\finelen),\epsilon)$-oblivious covering of $p$
for any distribution $p$ on $[0,1]^d$ and for all $k, d,\epsilon$ with $\epsilon \leq 1/2$.

\medskip

Our basic plan will be that if $p$ is a distribution with a $(k,j,\ell,\epsilon/2)$-oblivious covering $\mathcal{F}$,
and $q$ is a $k$-histogram that differs from $p$ by at least $\epsilon$ in $L_1$-distance, then $q$ defines a partition $\Pi$ of
$[0,1]^d$ into $k$ rectangles. This partition gives rise to a subfamily $\mathcal{S}\subseteq \mathcal{F}$
satisfying the constraints specified in Definition~\ref{defn:oblivious-covering}.
We would like to show that a constant fraction of the discrepancy between $p$ and $q$ can be detected
by considering their restrictions to $\mathcal{S}$. There are a couple of obstacles to showing this,
the first of which is that we do not know what $\mathcal{S}$ is. Fortunately, we do have the guarantee
that $|\mathcal{S}|$ is relatively small. We can consider the restrictions of $p$ and $q$ over all sets in $\mathcal{S}$
and try to check if there is a significant discrepancy between the two coming from any small subset.
To achieve this, we will make essential use of an identity tester under the $\ell^1_k$-metric, which
we now define:

\begin{definition}[$\ell_1^k$-distance]
Let $p$ and $q$ be distributions on a finite size domain, that we denote by $[n]$ without loss
od generality.
For any positive integer $k\geq 1$, we define $\|p-q\|_{1,k}$ as the sum of the largest $k$
values of $|p(i)-q(i)|$ over $i \in [n]$.
\end{definition}

Note that $\|p-q \|_{1, k} \geq \eps$ means that there exists a set $\mathcal{A}$ of $k$ or fewer
domain elements such that $\sum_{s \in \mathcal{A}} |p(s)-q(s)| \geq \eps$. That is,
these elements alone contribute at least $\eps$ to the $\ell_1$-distance between the distributions.

We start by proving the following theorem:

\begin{theorem}[Sample-Optimal $\ell_1^k$ Identity Testing]\label{thm:l1k}
Given a known discrete distribution $p$ and sample access to an unknown discrete distribution $q$,
each of any finite domain size, there exists an algorithm that accepts with probability $2/3$ if $p=q$ 
and rejects with probability $2/3$ if $\|p-q \|_{1, k} \geq \eps.$
The tester requires only knowledge of the known distribution $p$ and
$O(\sqrt{k} / \epsilon^2)$ samples from $q$.
\end{theorem}

Recall that if we wanted to distinguish between $p=q$ and $\|p-q\|_1 > \eps$, 
this would require $\Omega(\sqrt{n}/\eps^2)$ samples. However, the optimal $\ell_1$-identity 
testers are essentially adaptations of $\ell_2$-testers. That is, roughly speaking, they actually distinguish 
between $p=q$ and $\|p-q\|_2 > \epsilon/\sqrt{n}$. Hence, it should be intuitively clear why 
it would be easier to test for discrepancies in $\ell_1^k$-distance: 
If $\|p-q\|_{1,k} > \eps$, then $\|p-q\|_2 > \epsilon/\sqrt{k}$, making it easier for an  
$\ell_2$-type tester to detect the difference. \new{We apply the flattening technique of~\cite{DK16}
combined with the $\ell_2$-tester of~\cite{CDVV14} to obtain our optimal $\ell_1^k$-identity tester.
We note that an optimal $\ell_1^k$ closeness tester between discrete distributions 
was given in~\cite{DKN17}. The proof of Theorem~\ref{thm:l1k} follows along the same lines
and is given in Appendix~\ref{app:l1k}.}

The second obstacle is that although $q$ will be constant within each $S\in \mathcal{S}$,
it will not necessarily be the case that $p(S)$ and $q(S)$ will differ substantially even if the
variation distance between $p$ and $q$ on $S$ is large. To fix this, we show that $S$ can
be split into two parts such that at least one of the two parts will necessarily detect 
a large fraction of this difference:

\begin{lemma}\label{lem:split-capturing}
Let $p, q: \R^d \to \R_{+}$ and let $S$ be a bounded open subset in $\mathbb{R}^d$ on which $q$ is uniform.
Suppose $S$ is partitioned into two subsets $S_1,S_2$ such that $\vol(S_1)=\vol(S_2)=\vol(S)/2$ and
$p(s_1) \geq p(s_2)$ for all $s_1 \in S_1, s_2 \in S_2$, where $\vol()$ denotes Euclidean volume. Then,
$$
\max\left\{\left|\int_{S_1} (p(x)-q(x)) \mathrm{d} x\right|, \left|\int_{S_2}(p(x)-q(x)) \mathrm{d} x\right|\right\} \geq \int_{S} |p(x)-q(x)| \mathrm{d} x / 4.
$$
\end{lemma}
\begin{proof}
Let $W\subseteq S$ be the set of points $x \in S$ for which $p(x) \geq q(x)$ and $W'=S\backslash W$. 
Then we have that
$$
\int_{S} |p(x)-q(x)| \mathrm{d} x = \int_{W} (p(x)-q(x)) \mathrm{d} x + \int_{W'} (q(x)-p(x)) \mathrm{d} x.
$$
We will show that
\begin{equation}\label{eq:subsequent-error}
\max\left\{\left|\int_{S_1} (p(x)-q(x)) \mathrm{d} x\right|, \left|\int_{S_2} (p(x)-q(x)) \mathrm{d} x\right|\right\} \geq \int_{W} (p(x)-q(x)) \mathrm{d} x/2.
\end{equation}
By an argument analogous to the one we will give to prove Equation~\eqref{eq:subsequent-error}, one can also prove that
$$
\max\left\{\left|\int_{S_1} (p(x)-q(x)) \mathrm{d} x\right|, \left|\int_{S_2} (p(x)-q(x)) \mathrm{d} x\right|\right\} \geq \int_{W'} (q(x)-p(x)) \mathrm{d} x/2.
$$
Combining the above will give Lemma~\ref{lem:split-capturing}.

Note that if $S_1 = S \cap W$, Equation \eqref{eq:subsequent-error} immediately holds.
In fact, it holds even without the factor of two on the right hand side.
Similarly, if $S_1 \subseteq S \cap W$, then it also holds (but this time with the factor of two). To show this, 
note that
\[
\int_{W} (p(x)-q(x)) \mathrm{d} x  = \int_{S_1 \cap W} (p(x)-q(x)) \mathrm{d}x + \int_{S_2 \cap W} (p(x)-q(x)) \mathrm{d}x = \int_{S_1} (p(x)-q(x)) \mathrm{d}x + \int_{S_2 \cap W} (p(x)-q(x)) \mathrm{d}x.
\]
The RHS is a sum of two integrals where the second integral's integrand is always smaller 
than the smallest value of the first integral's integrand. Furthermore, the second integral 
is over a region that is no larger than the first region of the first integral, 
because $\vol(S_1)=\vol(S)/2$, while $\vol(S_2 \cap W) \leq \vol(S_2) = \vol(S)/2$. 
Thus, we have
\[
\int_{W} (p(x)-q(x)) \mathrm{d} x  \leq 2 \int_{S_1} (p(x)-q(x)) \mathrm{d}x \;,
\]
which implies Equation~\eqref{eq:subsequent-error}.

The final case needed to prove Equation~\eqref{eq:subsequent-error} holds is when $S_1 \cap W \subsetneq S_1$, which is equivalent to saying that $S_1$ contains points $x$ for which $p(x) < q(x)$. Let $h = -\int_{S_1 \cap W'} (p(x)-q(x)) \mathrm{d} x \geq 0$. Then we have
\[
\int_{W} (p(x)-q(x)) \mathrm{d} x = h + \int_{S_1} (p(x)-q(x)) \mathrm{d} x \;.
\]
If $h \leq \int_{W} (p(x)-q(x)) \mathrm{d} x/2$, then
we can substitute this into the preceding equation and we are done.
Otherwise, $h > \int_{W} (p(x)-q(x)) \mathrm{d} x/2$.
Note that  in this case, $|\int_{S_2} (p(x)-q(x)) \mathrm{d} x| \geq h$.
\footnote{This is because the integrand on the RHS is always more negative value of the integrand
on the RHS and the region the integral on the LHS is over is at least as large as that of the integral on the RHS.
This is very similar to the reasoning in the earlier case where $S_1 \cap W = S_1$ above.} Putting these together gives
\[
\left|\int_{S_2} (p(x)-q(x)) \mathrm{d} x \right| > \int_{W} (p(x)-q(x)) \mathrm{d} x/2 \;,
\]
completing the proof.
\end{proof}

We can now state the main algorithmic result of this section:

\begin{theorem}\label{thm:covering-implies-algo}
Let $p$ be a known distribution on $[0,1]^d$ with a $(k,j,\ell,\eps/2)$-\emph{oblivious covering}.
There exists a tester that given sample access to an unknown $k$-histogram $q$ on $[0,1]^d$ 
distinguishes between $p=q$ and $\dtv(p,q)\geq \eps$ with probability at least $2/3$
using $O(\sqrt{k j} \cdot \ell^2 / \epsilon^2)$ samples.
\end{theorem}

Plugging in the bounds on $j$ and $\ell$ of $2^{d} \log^{d}(kd/\epsilon)$ and $\log^{d}(kd/\epsilon)$ from Lemma~\ref{lem:oblivious-covering-construction} (established in Section~\ref{ssec:good-cover}) yields a sample complexity upper bound of  
$O(\sqrt{k} 2^{d/2}\log^{2.5 d}(kd/\epsilon) / \epsilon^2)$ for $\epsilon \leq 1/2$. 
This gives the upper bound portion of Theorem~\ref{thm:main}.

\smallskip

The high-level idea of the algorithm establishing Theorem~\ref{thm:covering-implies-algo} 
is to take each element of the oblivious cover and divide it in two, as in Lemma \ref{lem:split-capturing}, 
and then use the tester from Theorem \ref{thm:l1k} on the induced distributions of $p$ and $q$ on the resulting sets. 
The algorithm itself is quite simple and is presented in pseudo-code below.

\begin{algorithm}
\caption{Identity Tester for $d$-dimensional $k$-histograms}
\label{alg:tester}

Input: sample access to $k$-histogram distribution $q: [0, 1]^d \to \R_+$, $\eps > 0$, 
and explicit distribution $p: [0, 1]^d \to \R_+$ with $(k, j,\ell,\eps/2)$-oblivious covering.\\
Output: ``YES'' if $q = p$; ``NO'' if $\|q-p\|_{1} \ge \eps.$

\begin{enumerate}
\item Let $\mathcal{F}$ be a $(k,j,\ell,\epsilon/2)$-{oblivious covering} of the known distribution $p$.
\item Obtain a new family of sets $\mathcal{F'}$ by taking each $S\in \mathcal{F}$ and replacing it with the two sets $S_1$ and $S_2$ as defined in Lemma~\ref{lem:split-capturing}.
\item Define discrete distributions $p',q'$ over $\mathcal{F'}$ where a random sample, $x$, from $p'$ (resp. $q'$) is obtained by taking a random sample from $p$ (resp. $q$) and then returning a uniform random element of $\mathcal{F'}$ containing $x$. 
(We note that the distribution $p'$ can be explicitly computed, and we can take a sample from $q'$ 
at the cost of taking a sample from $q$.)
\item Use the algorithm from Theorem~\ref{thm:l1k} to distinguish between $p'=q'$ and the existence of a set $\mathcal{A}$ of size at most $2 k\cdot j$ with $\sum_{S\in \mathcal{A}} |p'(S)-q'(S)| \geq \eps/(8\ell)$.
\item Output ``YES'' in the former case and ``NO'' in the latter case.
\end{enumerate}
\end{algorithm}

\begin{proof}[Proof of Theorem~\ref{thm:covering-implies-algo}]
We note that the sample complexity of the tester described
in Algorithm~\ref{alg:tester} is $O(\sqrt{kj}\ell^2/\eps^2)$, as desired.
It remains to prove correctness.

The completeness case is straightforward.
If $p=q$, then clearly $p'=q'$ and our tester will accept with probability at least $2/3$.

We now proceed to prove soundness.
If $\dtv(p,q)\geq \eps$, we claim that our tester will reject with probability at least $2/3$.
For this we note that the unknown distribution $q$ defines some partition $\Pi$ of $[0,1]^d$ into $k$ rectangles
such that $q$ is constant on each part of the partition.
By the definition of an oblivious cover, there is a subfamily of disjoint sets $\mathcal{S}\subseteq \mathcal{F}$ such that:
\begin{itemize}
\item $q$ is constant on each element of $\mathcal{S}$.
\item $|\mathcal{S}|\leq k\cdot j$.
\item Letting $V=\bigcup_{S\in \mathcal{S}} S$, we have that $p\left(V\right) \geq 1-\eps/2$.
\end{itemize}
Since $\eps=\dtv(p,q)=\int_{[0,1]^d} \max(p-q,0)dx,$ we have that 
$\int_{V} \max(p-q,0)dx \geq \eps - \int_{[0,1]^d \setminus V}pdx \geq \eps/2$. Therefore, since the elements of $\mathcal{S}$ are disjoint, we have that
$\sum_{S\in \mathcal{S}} \int_S |p-q|dx \geq \eps/2.$

We now let $\mathcal{A}\subseteq\mathcal{F'}$ be the collection of all $S_1$ or $S_2$ corresponding to an $S\in \mathcal{S}$. We note that $|\mathcal{A}|=2|\mathcal{S}| \leq 2k\cdot j$. Furthermore, by Lemma \ref{lem:split-capturing}, we have that
\begin{align*}
\eps/8 & \leq \sum_{S\in \mathcal{S}} \int_S |p-q|dx/4\\
& \leq \sum_{S\in \mathcal{S}} \max \left\{ |p(S_1)-q(S_1)|, |p(S_2)-q(S_2)| \right\}\\
& \leq \sum_{A\in \mathcal{A}} |p(A)-q(A)|.
\end{align*}
On the other hand, for $A\in\mathcal{A}$, we have that $p'(A) = p(A)/\ell$ and $q'(A)=q(A)/\ell$, so we have that
$$
\sum_{A\in\mathcal{A}} |p'(A)-q'(A)| \geq \eps/(8\ell).
$$
Therefore, if $\dtv(p,q)\geq \eps$, our tester will reject with probability at least $2/3$.

This completes the proof of Theorem~\ref{thm:covering-implies-algo}.
\end{proof}

\begin{remark} \label{rem:rob}
{\em Algorithm~\ref{alg:tester} is robust in the sense that it still works even if $q$ is only (say) $\epsilon/10$-close to some $k$-histogram distribution $\widetilde{q}$ instead of actually being one. To show this, one can note that the existing proof applied to $p$ and $\widetilde{q}$ gives an $\mathcal{A}$ such that $\sum_{A \in \mathcal{A}} |p(A) - \widetilde{q}(A)|$ is at least $\epsilon/8$. The triangle inequality then implies $\sum_{A \in \mathcal{A}} |p(A) - q(A)| \geq \epsilon/40$, which, by the same reasoning given in the proof of the non-robust case, implies the algorithm is still correct.}
\end{remark}

\begin{remark} \label{rem:disc-ub}
{\em Even though our testing algorithm was phrased for histograms
over $[0, 1]^d$, it can be made to apply for discrete histograms on $[m]^d$ via a simple reduction. 
In particular, if each element of $[m]^d$ is replaced by a box of side length $1/m$ on each side, 
a $k$-histogram on $[m]^d$ is transformed into a $k$-histogram over $[0,1]^d$, 
in a way that preserves total variation distance. If our algorithm is applied to the latter
histogram, we can obtain correct results for the former.}
\end{remark}

\subsection{Construction of Good Oblivious Covering} \label{ssec:good-cover}

In this section, we prove the existence of an oblivious covering:

\begin{lemma}\label{lem:oblivious-covering-construction}
For any continuous distribution $p$ on $[0,1]^d$, positive integer $k$ and $\eps\leq 1$,
there exists a $\left(k, 2^d \log^{d}(\finelen), \log^{d}(\finelen), \eps \right)$-oblivious covering of $p$.
\end{lemma}

\begin{proof}
The basic idea of our construction will be to let $\mathcal{F}$ be a union of
grids where the number of cells in each direction is a power of $2$.

For each coordinate, $j \in [d]$, and each non-negative integer $i$, define the $i^{th}$ partition
of this coordinate to be a partition of $[0,1]$ into $2^i$ intervals such that $j^{th}$ marginal, $p_j$, 
of $p$ assigns each interval in the partition equal mass, and such that the $i^{th}$ partition is a refinement of the $(i-1)^{st}$.

\begin{figure}
    \centering
    \begin{subfigure}[b]{0.2\textwidth}
        \includegraphics[width=\textwidth]{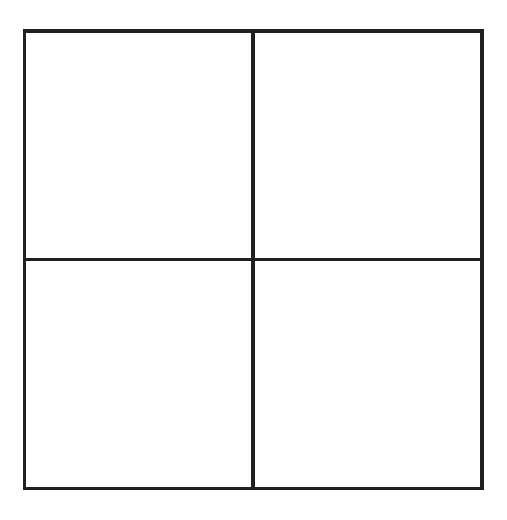}
        \caption{A $z$-grid with $z=[1,1]$.}
    \end{subfigure}
\quad
    \begin{subfigure}[b]{0.2\textwidth}
        \includegraphics[width=\textwidth]{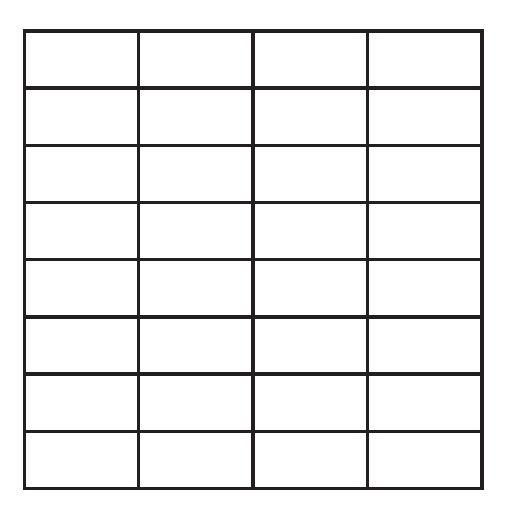}
        \caption{A $z$-grid with $z=[2,3]$.}
    \end{subfigure}
    \caption{$z$-grids for different values of $z$ partition $[0,1]^2$ into rectangles. In this figure, the axes are scaled such that the marginal distributions of the vertical and horizontal coordinates, respectively, of $p$ are uniform.} \label{fig:z-grid}
\end{figure}


For each vector $z \in \mathbb{N}^d$, define the $z$-grid as the partition of $[0,1]^d$
into rectangles by taking the product of the $z_j^{th}$ partition of the $j^{th}$ coordinate.
We let $\mathcal{F}$ be the union of the cells in the $z$-grid for all $z \in \mathbb{N}^d \cap [0,m-1]$
for $m=\log_2(\finelen)$. An illustration is given in Figure~\ref{fig:z-grid}.

We note that each $x\in [0,1]^d$ is in exactly one cell in each $z$-grid,
and therefore is contained in exactly $m^d$ elements of $\mathcal{F}$, verifying Property 2.

For Property 1, consider a partition of $[0,1]^d$ into rectangles $R_1,\ldots,R_k$.
We claim that for each $R_i$ there is a subfamily $\mathcal{T}_i\subseteq \mathcal{F}$
of disjoint subsets of $R_i$ with $|\mathcal{T}_i|\leq 2^d m^d$,
and such that $p\left( R_i \backslash \bigcup_{S\in \mathcal{T}_i} S\right) \leq \eps/k$.
It is then clear that taking $\mathcal{S}$ to be the union of the $\mathcal{T}_i$ will suffice.
In fact, we will show that for any rectangle $R_i$, there is a corresponding $\mathcal{T}_i$ with these properties.
\begin{figure}
    \centering
        \includegraphics[width=2in]{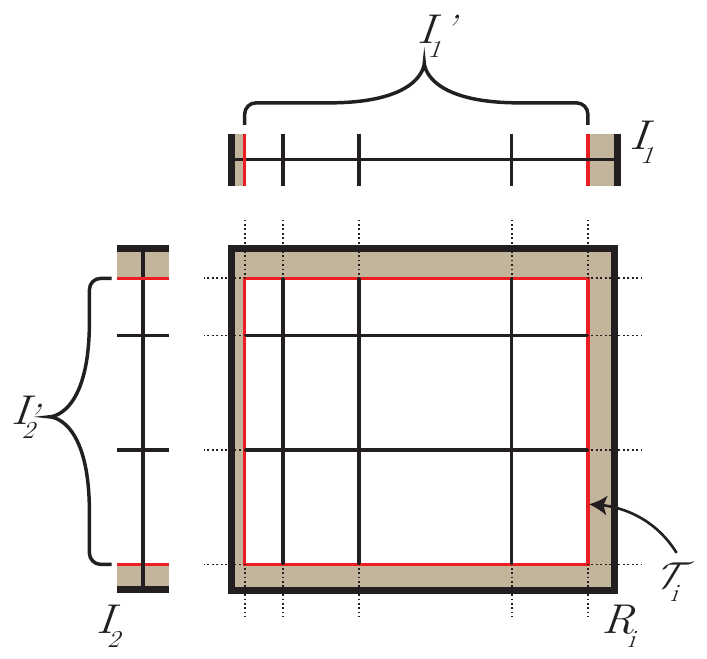}
    \caption{How our oblivious covering is used to cover a rectangle $R_i$ in the proof of Lemma~\ref{lem:oblivious-covering-construction}. Each dimension of $R_i$ is separately decomposed into non-overlapping one-dimensional rectangles, with a small amount of area shaded in beige left over on the sides. $\mathcal{T}_i$ is obtained by taking the family of all Cartesian products of the form $I''_1 \times \cdots \times I''_d$ where, for each $j$, $I''_j$ is any subinterval in the decomposition of $I'_j$. In this figure, the axes are scaled such that the marginal distributions of the vertical and horizontal coordinates, respectively, of $p$ are uniform.} \label{fig:covering}
\end{figure}

We let $R_i=\prod_{j=1}^d I_j$ for intervals $I_j$. We let $I_j'$ be $I_j$ minus the intervals of the $(m-1)^{st}$-partition
of the $j^{th}$ coordinate that contain the endpoints of $I_j$. We note that $p_j(I_j\backslash I_j') \leq \eps/(kd)$
and that $I_j'$ is a union of consecutive intervals in the $(m-1)^{st}$ partition of this coordinate.
We claim that this means that $I_j'$ is the union of at most $2m$ intervals of one of the first $m-1$ partitions
of the $j^{th}$ coordinate. This is easy to see by induction on $m$, as $I_j'$ is a union of consecutive intervals
in the $(m-2)^{nd}$ partition union at most one interval of the $(m-1)^{st}$ on either end.
The one-dimensional intervals on the top and left of Figure~\ref{fig:covering} show an illustration of this.

In order to produce $\mathcal{T}_i$, we write each $I_j'$ as a union of at most $2m$ intervals
from the relevant partitions. We let $\mathcal{T}_i$ be the set of rectangles obtained by taking
the product of one rectangle from each of these sets. It is then clear that $\mathcal{T}_i$
partitions $\prod_{j=1}^d I_j'$ into at most $(2m)^d$ pieces. Figure~\ref{fig:covering}
shows an illustration of this. We now note that
$$
p\left( R\backslash \prod_{j=1}^d I_j'\right) = p\left( \prod_{j=1}^d I_j\backslash \prod_{j=1}^d I_j'\right) \leq \sum_{j=1}^d p_j(I_j\backslash I_j') \leq \eps/k \;.
$$
Thus, $\mathcal{T}$ satisfies all of the desired properties, and taking the union of the $\mathcal{T}_i$
will yield an appropriate $\mathcal{S}$. This completes the proof of Lemma~\ref{lem:oblivious-covering-construction}.
\end{proof}

\section{Sample Complexity Lower Bound} \label{sec:lb}

In this section, we prove the sample complexity lower bound of Theorem~\ref{thm:main}.
The structure of this section is as follows:
We begin (Proposition~\ref{basicLBProp}) by providing a new proof 
that $\Omega(\sqrt{k}/\eps^2)$ samples are required to test uniformity of a $k$-histogram in one dimension. 
The purpose of reproving this previously known result is so that we may later generalize it to higher dimensions.
We then proceed by describing a basic construction that yields a slightly improved lower bound 
(Proposition~\ref{IntermediateLBProp}).
Finally, we present a more sophisticated construction that suffices to establish our final lower bound 
in Theorem \ref{finalLBThm}.

\paragraph{Basic Background.}
Recall the definition of the $\chi$-metric. 
Notice that, for fixed $q$, $\chi_p(q,r)$ is an inner product on distributions $q, r$.
Furthermore, by the Cauchy-Schwarz inequality it follows that if $q$ and $p$ are probability distributions then
$$
\chi_p(q,q) = \int\frac{dq^2}{dp} = \left( \int\frac{dq^2}{dp}\right) \left( \int dp \right) \geq \left( \int dq \right)^2 = 1 \;.
$$
This metric is useful for determining whether or not distributions can be distinguished.
In particular, if $q$ and $p$ can be distinguished from a single sample,
it must be the case that $\chi_p(q,q)$ is much bigger than $1$. Formally, we have:

\begin{lemma}\label{chiSquaredLBLem}
Suppose that $q$ and $p$ are probability distributions.
Suppose furthermore that there is an algorithm that given a random sample from $q$ accepts
with probability at least $2/3$, and given a random sample from $p$ rejects with probability at least $2/3$.
Then, it holds that $\chi_p(q,q) \geq 4/3$.
\end{lemma}
\begin{proof}
Let $A$ be the set on which the algorithm accepts. 
We then have that $q(A)\geq 2/3$ and $p(A) \leq 1/3$. Therefore, we have that
$$
\chi_p(q,q) \geq \int_A \frac{dq^2}{dp} \geq 3\left( \int_A\frac{dq^2}{dp}\right) \left( \int_A dp \right) \geq 3\left( \int_A dq \right)^2\geq 4/3.
$$
\end{proof}

\subsection{Lower Bound for Uniformity Testing of Univariate Histograms} \label{ssec:lb-1d}

We start by using Lemma~\ref{chiSquaredLBLem} to prove a lower bound 
on the number of samples required to test uniformity of univariate $k$-histograms.
We build on this argument in the following subsections to establish our final multidimensional
lower bound.

The idea is to use a standard adversary argument, using Lemma \ref{chiSquaredLBLem} to show that it is impossible
to distinguish samples taken from a distribution from a particular ensemble, from those taken from the uniform distribution.
\begin{proposition}\label{basicLBProp}
If there exists an algorithm that given $s$ independent samples from an unknown $k$-histogram, $q$, on $[0,1]$
and accepts with at least $2/3$ probability if $q=U$ and rejects with at least $2/3$ probability if $\dtv(q,U)\geq \eps$,
then $s = \Omega(\sqrt{k}/\eps^2)$.
\end{proposition}
\begin{proof} 
We assume that $k$ is even. Divide $[0,1]$ into $k/2$ equally sized bins. 
Let $\mathcal{P}$ be a distribution over $k$ histograms
where in each bin either $dq=(1+\eps)dx$ on the first half and $dq=(1-\eps)dx$ on the second half of the bin,
or visa versa independently for each bin. Note that a sample from $\mathcal{P}$ is always a $k$-histogram $q$
with $\dtv(q,U)=\eps$. Let $\mathcal{P}^{\etens s}$ be the distribution on $[0,1]^s$ obtained by randomly picking a distribution $q$
from $\mathcal{P}$ and then taking $s$ independent samples from $q$.

Given that an algorithm to distinguish the uniform distribution from $k$-histograms far from it exists,
such a distribution can distinguish a single sample from $\mathcal{P}^{\etens s}$ from a sample from $U^{\etens s}$.
Therefore, by Lemma \ref{chiSquaredLBLem}, we must have that $\chi_{U^{\etens s}}(\mathcal{P}^{\etens s},\mathcal{P}^{\etens s}) \geq 4/3.$
We will now try to bound this quantity.

Note that $\mathcal{P}^{\etens s}$ is a mixture of the distributions $q^{\etens s}$ where $q$ is drawn from $\mathcal{P}$.
Therefore, by linearity of the $\chi$-metric, we have that
$$
\chi_{U^{\etens s}}(\mathcal{P}^{\etens s},\mathcal{P}^{\etens s}) = \E_{p, q \sim  \mathcal{P}} [\chi_{U^{\etens s}}(p^{\etens s}, q^{\etens s})] = 
\E_{p, q \sim  \mathcal{P}} [(\chi_{U}(p,q))^s] \;,$$
where the last equality is by noting that the corresponding integral decomposes as a product.

We now need to think about the distribution of $\chi_U(p, q)$
when $p$ and $q$ are drawn independently from $\mathcal{P}$.
We note that for each bin $B$ the quantity $\int_B \frac{dp dq}{dU}$ is either
$\frac{1+\eps^2}{k/2}$ or $\frac{1-\eps^2}{k/2}$ with equal probability and independently for each bin.
Therefore,
$$
\chi_{U^{\etens s}}(p^{\etens s}, q^{\etens s}) \sim \left(1+\frac{\eps^2}{k/2} \sum_{i=1}^{k/2} X_i \right)^s \;,
$$
where $X_i$ are i.i.d. random variables $X_i \in_u \{\pm 1\}$. Therefore,
$$
\chi_{U^{\etens s}}(\mathcal{P}^{\etens s},\mathcal{P}^{\etens s}) = \E\left[\left(1+\frac{\eps^2}{k/2} \sum_{i=1}^{k/2} X_i \right)^s \right] \;.
$$
To bound this quantity, we use the fact that, for each $t$, the $t^{th}$ moment of a Rademacher random variable
is less than or equal to the corresponding moment of the standard Gaussian. We thus have that
$$
\chi_{U^{\etens s}}(\mathcal{P}^{\etens s},\mathcal{P}^{\etens s}) \leq \E\left[\left(1+\frac{\eps^2}{k/2} \sum_{i=1}^{k/2} G_i \right)^s \right] \;,
$$
where the $G_i$ are i.i.d. $N(0,1)$ random variables. We can bound this latter quantity as follows:
\begin{align*}
\chi_{U^{\etens s}}(\mathcal{P}^{\etens s},\mathcal{P}^{\etens s}) & \leq \E\left[\left(1+\frac{\eps^2}{\sqrt{k/2}} N(0,1) \right)^s \right]\\
& \leq \E\left[ \exp\left(\left(\frac{s\eps^2}{\sqrt{k/2}}\right) N(0,1)  \right) \right]\\
& = \exp\left(\left(\frac{s\eps^2}{\sqrt{k/2}}\right)^2/2 \right) \;.
\end{align*}
Hence, a testing algorithm can only exist when
$$
\left(\frac{s\eps^2}{\sqrt{k}}\right) \geq \sqrt{\log(4/3)} \;,
$$
or equivalently when $s = \Omega(\sqrt{k}/\eps^2).$
This completes the proof of Proposition~\ref{basicLBProp}.
\end{proof}

\subsection{First Attempt: Basic Multidimensional Lower Bound} \label{ssec:lb-d-first}
In this subsection, we build on the univariate construction of the previous
subsection to obtain a slightly improved lower bound in $d$ dimensions. 
We achieve this by modifying our ensemble in order to force any testing algorithm
to guess the dimensions of the rectangles involved in the partition.
Specifically, we prove the following:

\begin{proposition}\label{IntermediateLBProp}
If there exists an algorithm that, given $s$ independent samples from a $k$-histogram, $q$, on $[0,1]^d$
with $k>4^d$, accepts with at least $2/3$ probability if $p=U$ and rejects with at least $2/3$ probability
if $\dtv(p,U)\geq \eps$, then $s = \Omega(\eps^{-2}\sqrt{kd/2^d\log(\log(k-d)/d)})$.
\end{proposition}
\begin{proof}
\begin{figure}
    \centering
        \includegraphics[width=0.3\textwidth]{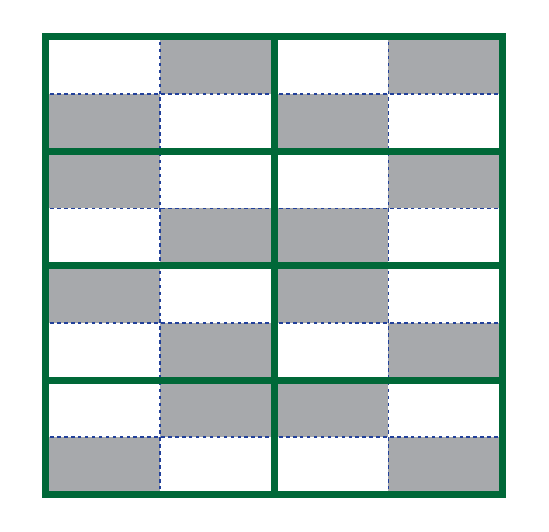}
\caption{An example of a distribution from $\mathcal{P}$. The dark cells have density $1+\eps$, and the light cells have density $1-\eps$. The green lines separate the square into a $4\times 2$ grid, and each rectangle is filled with a random $2\times 2$ checkerboard.} \label{fig:lb-grid}
\end{figure}
We first assume that $k$ is a power of $2$, namely $k=2^{m+d}$.
Since this can always be achieved by decreasing $k$ by a factor of at most $2$, this should not affect the final bound.
We define an ensemble $\mathcal{P}$ similarly to how we did so in the proof of Proposition \ref{basicLBProp}.
To define a distribution $q$ in $\mathcal{P}$, first we randomly and uniformly pick a $d$-tuple $(m_1,m_2,\ldots,m_d)$
of non-negative integers summing to $m$. We call this {\em the defining vector} of $q$. We next divide $[0,1]^d$
into $k/2$ bins by producing a \new{$\prod_{j=1}^d 2^{m_j}$} grid 
We cut each bin into $2^d$ equal sub-bins by diving it in half along each dimension. 
We divide these sub-bins into two classes based on their parity.
We then let $dq=(1+\eps)dV$ on the sub-bins of a random parity and $dq=(1-\eps)dV$ on the other sub-bins,
where the choices are independent for each bin. We note that a $q$ drawn from $\mathcal{P}$ 
is always a $k$-histogram that is $\eps$-far from the uniform distribution $U$. An illustration is given in Figure~\ref{fig:lb-grid}.

We let $\mathcal{P}^{\etens s}$ be the distribution on $([0,1]^d)^s$ obtained by drawing a random $q$ from $\mathcal{P}$
and taking $s$ independent samples from $q$. Once again, it suffices to bound from below 
$\chi_{U^{\etens s}}(\mathcal{P}^{\etens s},\mathcal{P}^{\etens s})$.
We similarly have that
$$
\chi_{U^{\etens s}}(\mathcal{P}^{\etens s},\mathcal{P}^{\etens s}) = \E_{p,q \sim  \mathcal{P}} [(\chi_{U}(p,q))^s] \;.
$$
We note that if $p$ and $q$ have the same defining vectors, then the contribution to $\chi_U(p,q)$
from each bin is randomly and independently $2(1\pm \eps^2)/k$. Therefore, by the arguments 
of the previous subsection, if we condition on $p$ and $q$ having the same defining vectors, 
the expectation of $(\chi_{U}(p,q))^{\etens s}$ is
at most $\exp\left(\left(\frac{s\eps^2}{\sqrt{k/2^d}}\right)^2 \right)$. On the other hand, if $p$ and $q$ have
different defining vectors, we claim that $\chi_U(p,q)=1$. In fact, we make the stronger claim that
if $A$ is the intersection of a defining bin of $p$ and a defining bin of $q$, then
$\int_A \frac{dp dq}{dU} = q(A)$. This is because without loss of generality we may assume that $p$'s
associated $m_1$ is smaller than $q$'s associated $m_1$. This in turn means that given any point in $A$,
the entire width of $A$ along the first axis will be in the same sub-bin for $q$, but will pass through
two sub-bins of opposite parity for $p$. Thus, the average of $dp/dU$ over this line will be $1$,
and thus the integral over $A$ of $dp dp/dU$ is the same as the integral of $dq$.

Now since there are $\binom{m+d-1}{d-1}$ different possible defining vectors, we have that
$$
\chi_{U^{\etens s}}(\mathcal{P}^{\etens s},\mathcal{P}^{\etens s}) \leq 1 + \binom{m+d-1}{d-1}^{-1} \exp\left(\left(\frac{s\eps^2}{\sqrt{k/2^d}}\right)^2 \right) \;.
$$
In order for this to be at least $4/3$, it must be the case that
$$
\left(\frac{s\eps^2}{\sqrt{k/2^d}}\right) \gg \sqrt{\log \binom{m+d-1}{d-1}} \;,
$$
or
$$
s = \Omega(\eps^{-2}\sqrt{kd/2^d\log(\log(k-d)/d)}) \;.
$$
This completes the proof of Proposition~\ref{IntermediateLBProp}.
\end{proof}

\subsection{Second Attempt: Proof of Final Sample Lower Bound} \label{ssec:lb-d-final}

Unfortunately, the lower bound of Proposition~\ref{IntermediateLBProp} only saves us a $\log\log(k)$ factor.
This is essentially because a testing algorithm only needs to correctly guess one of poly-logarithmically
many defining vectors, and once it has guessed the correct one, it only needs to see a signal large
enough that the probability of error is only inverse poly-logarithmic. This can be done by increasing
the number of samples by only a doubly logarithmic factor. In order to do better, we will need
a slightly more complicated construction, where we chop our domain into pieces and fill each piece
with rectangles, but where different pieces might have rectangles of different sizes.
\begin{theorem}\label{finalLBThm}
If there exists an algorithm that, given $s$ independent samples from a $k$-histogram, $q$, on $[0,1]^d$
with $k>2^{100d}$, accepts with at least $2/3$ probability if $q=U$, and rejects with at least $2/3$ probability
if $\dtv(p,U)\geq \eps$, then $s = (\sqrt{k}/\eps^2) \cdot \Omega(\log(k)/d)^{d-1}$.
\end{theorem}
\begin{proof}

\begin{figure}
    \centering
        \includegraphics[width=0.3\textwidth]{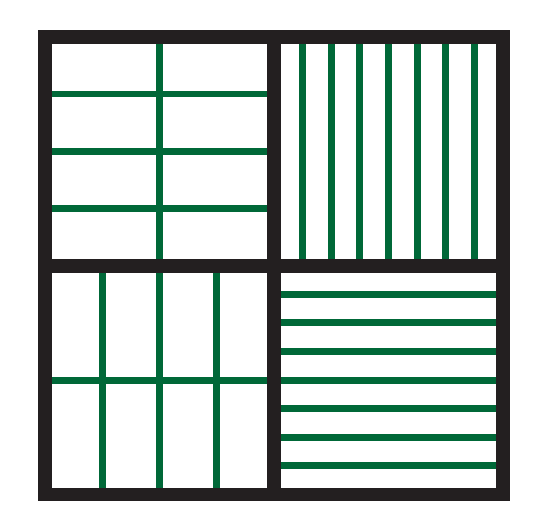}
\caption{An example of a probability distribution from ensemble $\mathcal{Q}$. The square is divided into $n=4$ regions by the black lines. Each sub-square is divided into a randomly sized grid of $2^m=8$ equal rectangles by the green lines. To get the final distribution, each of those rectangles should be filled with a random checkerboard as in Figure~\ref{fig:lb-grid}.}
\label{fig:lb-grid-strong}
\end{figure}

We first assume that $k$ can be written in the form $k=n2^{m+d}$,
where $n \leq \binom{m+d-1}{d-1}/4$. We note that (perhaps decreasing $k$ by a constant factor)
we can achieve this with $n = \Omega(\log(k)/d)^d$, and therefore we can assume this throughout the rest of the argument.

We describe a new ensemble $\mathcal{Q}$ over $k$-histograms on $[0,1]^d$ in the following way:
First, divide $[0,1]^d$ into $n$ equal volume boxes in some arbitrary way.
For each box $B_i$, pick a member $p_i$ from $\mathcal{P}$, the ensemble from the proof of
Proposition \ref{IntermediateLBProp}, independently for different $i$. We let the restriction of $q$ to $B_i$
be $p_i$ rescaled such that it assigns $B_i$ total mass $1/n$, and such that the domain of definition is
$B_i$, rather than $[0,1]^d$. An example element of $\mathcal{Q}$ is illustrated in Figure~\ref{fig:lb-grid-strong}.

Similarly, it suffices to show that if $s$ is below our desired sample lower bound then
$$
\chi_{U^{\etens s}}(\mathcal{Q}^{\etens s},\mathcal{Q}^{\etens s}) = \E_{p,q\sim  \mathcal{Q}} [(\chi_{U}(p,q))^s]
$$
is less than $4/3$.

We note that for $p$ and $q$ drawn from $\mathcal{Q}$ the quantity
$\int_{B_i} \frac{dp dq}{dU}$ is distributed as $\chi_U(p',q')/n$ with $p'$ and $q'$
drawn from $\mathcal{P}$. This is $1/n$ except with probability $\alpha:=\binom{m+d-1}{d-1}^{-1}$
and otherwise is distributed as $1/n+\frac{\eps^2}{n 2^m}\sum_{j=1}^{2^m}X_{ij}$,
where the $X_{ij}$ are i.i.d. $\{\pm 1\}$ random variables. Notice that these are independent
for different $i$ and sum to $\chi_U(p,q)$. Therefore,
$$
\chi_U(p,q) \sim 1 + \sum_{i=1}^n Y_i \left(\frac{\eps^2}{n 2^m}\sum_{j=1}^{2^m}X_{ij} \right) \;,
$$
where the $Y_i$ are i.i.d., equal to $1$ with probability $\alpha$, and $0$ otherwise.
Therefore, we have that
$$
\chi_{U^{\etens s}}(\mathcal{Q}^{\etens s},\mathcal{Q}^{\etens s}) = \E\left[ \left(1 + \sum_{i=1}^n Y_i \left(\frac{\eps^2}{n 2^m}\sum_{j=1}^{2^m}X_{ij} \right) \right)^s\right] \;.
$$
Once again, this expectation is only increased if the $X_{ij}$ are replaced by standard Gaussians, and so 
$\chi_{U^{\etens s}}(\mathcal{Q}^{\etens s},\mathcal{Q}^{\etens s})$ is at most
$$
\E\left[ \left(1 + \sum_{i=1}^n Y_i \left(\frac{\eps^2}{n 2^{m/2}}G_i \right) \right)^s\right]
$$
with $G_i$ i.i.d. standard normals. Noting that we still have a sum of 
$\sum_{i=1}^n Y_i \sim \mathrm{Binomial}(n,\alpha)$ many independent Gaussians,
this simplifies to
\begin{align*}
\chi_{U^{\etens s}}(\mathcal{Q}^{\etens s},\mathcal{Q}^{\etens s}) & \leq \E\left[ \left(1 + \left(\frac{\eps^2\sqrt{\mathrm{Binomial}(n,\alpha)}}{n 2^{m/2}}N(0,1) \right) \right)^s\right]\\
& \leq \E\left[\exp\left(\left(\frac{s\eps^2\sqrt{\mathrm{Binomial}(n,\alpha)}}{n 2^{m/2}} \right)^2/2 \right) \right]\\
& = \E\left[\exp\left(\mathrm{Binomial}(n,\alpha)\left(\frac{s^2\eps^4}{2 n^2 2^{m}} \right) \right) \right]\\
& \leq \left(1+\alpha \exp\left(\frac{s^2\eps^4}{2 n^2 2^{m}} \right) \right)^n\\
& \leq \exp\left(n\alpha\exp\left(\frac{s^2\eps^4}{2 n^2 2^{m}} \right) \right)\\
& \leq \exp\left(\exp\left(\frac{s^2\eps^4}{2 n^2 2^{m}} \right)/4 \right) \;.
\end{align*}
In order for this to be at least $4/3$, it must be the case that
$$
\frac{s^2\eps^4}{2 n^2 2^{m}} \gg 1 \;,
$$
or equivalently that
$$
s = \Omega(2^{m/2} n /\eps^2) = \Omega(\sqrt{kn/2^d}/\eps^2) = \Omega(\log(k)/d)^d\sqrt{k}/\eps^2 \;.
$$
This completes the proof of Theorem~\ref{finalLBThm}.
\end{proof}

\begin{remark} \label{rem:dis-lb}
{\em We note that our lower bounds for uniformity testing of histograms on $[0,1]^d$ 
can be made to work for histograms on $[m]^d$, assuming that $m \gg k$. 
In particular, our lower bound construction requires first dividing our domain 
into $n$ equal boxes, and then subdividing each of these boxes into $k/n$ equal boxes 
in such a way that the number of subdivisions in each dimension is a power of $2$.
For simplicity, let us assume that $n$ is a power of $d$. In that case, we
can first cut each edge of our original box into $n^{1/d}$ equal pieces
and then further subdivide each side into $k/n$ equal pieces. We note
that all of the histograms in our adversarial family are consistent
with this partition of our cube into fewer than $k^d$ boxes. Therefore,
by the inverse of the reduction above, our lower bound can be
made to work on $[k]^d$ rather than $[0,1]^d$. A more elaborate construction 
can show that our lower bounds apply for domain $[m]^d$ for any $m \gg k$.}
\end{remark}

\section{Conclusions and Future Directions} \label{sec:conc}

In this work, we gave a computationally efficient and sample near-optimal algorithm
for the problem of testing the identity of multidimensional histogram distributions \new{in any fixed dimension}. 
Our nearly matching upper and lower bounds have interesting consequences regarding the relation of learning 
and identity testing for this important nonparametric family of distributions.

A natural direction for future work is to generalize our results to the problem of testing equivalence 
between two unknown multidimensional histograms. 
The one-dimensional version of this problem
was resolved in~\cite{DKN:15:FOCS, DKN17}.
Additional ideas are required for this setting, as the algorithm and analysis in this work 
exploit the a priori knowledge of the explicit distribution. 

Another direction for future work concerns characterizing the sample and computationally
complexity of identity testing $d$-dimensional $k$-histograms when the dimension $d$ and the number
of rectangles $k$ are comparable, e.g., $k = \poly(d)$ or even $k<d$. We believe that understanding 
these parameter regimes requires different ideas.


\bibliographystyle{alpha}
\bibliography{allrefs}

\appendix 
\section*{Appendix}

\section{Proof of Theorem~\ref{thm:l1k}} \label{app:l1k}
We use the flattening method developed in \cite{DK16}. 
We begin by giving the definition of a split distribution from that work:

\begin{definition} \label{def:split-distr}
Given a distribution $p$ on $[n]$ and a multiset $S$ of elements of $[n]$, 
define the \emph{split distribution} $p_S$ on $[n+|S|]$ as follows:
For $1\leq i\leq n$, let $a_i$ denote $1$ plus the number of elements of $S$ that are equal to $i$.
Thus, $\sum_{i=1}^n a_i = n+|S|.$ We can therefore associate the elements of $[n+|S|]$ to elements of the set
$B=\{(i,j):i\in [n], 1\leq j \leq a_i\}$.
We now define a distribution $p_S$ with support $B$, by letting a random sample from $p_S$ be given by $(i,j)$,
where $i$ is drawn randomly from $p$ and $j$ is drawn randomly from $[a_i]$.
\end{definition}

We recall a basic fact about split distributions:
\begin{fact}[Fact~2.5, \cite{DK16}]\label{splitDistributionFactsLem}
Let $p$ and $q$ be probability distributions on $[n]$, and $S$ be a given multiset of $[n]$. Then:
(i) We can simulate a sample from $p_S$ or $q_S$ by taking a single sample from $p$ or $q$, respectively.
(ii) It holds $\|p_S-q_S\|_1 = \|p-q\|_1$.
\end{fact}

We also recall an optimal $\ell_2$-closeness tester under the promise that one of the 
distributions has small $\ell_2$-norm:
\begin{lemma}[\cite{CDVV14}] \label{L2TestLem}
Let $p$ and $q$ be two unknown distributions on $[n]$.
There exists an algorithm that on input $n$,  $b \geq \min \{\|p\|_2, \|q\|_2 \}$
and $0< \eps < \sqrt{2}b$, 
draws $O(b/\eps^2)$ samples
from each of $p$ and $q$ and, with probability at least $2/3$,
distinguishes between the cases that $p=q$ and $\|p-q\|_2 > \eps.$
\end{lemma}

We now have all the necessary tools to describe and analyze our $\ell_1^k$-identity tester.
The pseudo-code of our algorithm follows:

\begin{algorithm}
\caption{\texttt{$\ell_1^k$-Identity-Tester}}
\label{alg:L1k-tester}
Input: sample access to discrete distribution $q: [n] \to [0, 1]$, $k \in \Z_+$, and $\eps > 0$, 
and explicit distribution $p: [n] \to [0, 1]$.\\
Output: ``YES'' if $q = p$; ``NO'' if $\|q-p\|_{1,k} \ge \eps.$
\begin{enumerate}

\item Let $S$ be the multiset obtained by taking $\lfloor k p_i \rfloor$ copies of $i \in [n]$.

\item Use the $\ell_2$-tester of Lemma~\ref{L2TestLem} to distinguish between the cases that 
$p_S=q_S$ and $\|p_S-q_S\|_2^2 \geq \eps^2/(2k)$ and return the result.

\end{enumerate}
\end{algorithm}

\medskip

We now provide the simple analysis.
Note that $|S| \leq \sum_{i=1}^n k p_i =k$ and that $p_S$ assigns probability mass at most $1/k$
to each domain element. Therefore, we have that $\|p_S\|_2  \leq 1/\sqrt{k}.$
By Lemma~\ref{L2TestLem} --- applied for $b = 1/\sqrt{k}$ and $\eps/\sqrt{2k}$ in place of $\eps$ --- 
we obtain that the $\ell_2$-tester in Step~2 of the above pseudo-code 
requires $O(b/\eps^2)  = O(\sqrt{k}/\eps^2)$ samples from $q_S$ and $p_S$.
Since $p$ is explicitly given, so is $p_S$ and therefore we can straightforwardly generate samples from
$p_S$ for free. By Fact~\ref{splitDistributionFactsLem}, we can generate a sample from $q_S$ given a sample from $q$.
Hence, our algorithm uses $O(\sqrt{k}/\eps^2)$ samples from $q$. This completes the analysis of the
sample complexity.

We now prove correctness. If $p=q$, then by Fact~\ref{splitDistributionFactsLem} we have that 
$p_S=q_S$ and the algorithm will return ``YES'' with appropriate probability. 
On the other hand, if $\|q-p\|_{1,k} \ge \eps$, then by definition of the $\ell_1^k$ metric it follows that 
$\|p_S-q_S\|_{1,k+m}\geq \eps$, for $m \eqdef |S|$. 
Since $k+m$ elements contribute to total $\ell_1$-error at least $\eps$,
by the Cauchy-Schwarz inequality, we have that 
$$\|p_S-q_S\|_2^2 \geq \epsilon^2/(k+m) \geq \eps^2/(2k) \;,$$
where we used the fact that $m = |S| \leq k$.
Therefore, in this case, the algorithm returns ``NO'' with appropriate probability.
This completes the proof of Theorem~\ref{thm:l1k}.

\end{document}